\documentclass[english,twocolumn]{IEEEtranTCOM}
\usepackage[T1]{fontenc}
\usepackage[latin9]{inputenc}
\usepackage{geometry}
\usepackage{amsthm}
\usepackage{amsmath}
\usepackage{amssymb}
\usepackage{graphicx}
\usepackage{esint}

\makeatletter
\numberwithin{equation}{section}
\numberwithin{figure}{section}
\theoremstyle{plain}
\newtheorem{thm}{\protect\theoremname}

\IEEEoverridecommandlockouts

\makeatother

\usepackage{babel}
\providecommand{\theoremname}{Theorem}

\begin{document}

\title{Spread Spectrum Codes for\\
Continuous-Phase Modulated Systems}

\author{Gaurav Thakur\IEEEauthorrefmark{1}%
\thanks{\IEEEauthorrefmark{1}MITRE Corporation, McLean, VA 22102, email:
gthakur@alumni.princeton.edu . Approved for Public Release; Distribution
Unlimited. 13-0855%
}}

\date{October 26, 2013}
\maketitle
\begin{abstract}
We study the theoretical performance of a combined approach to demodulation
and decoding of binary continuous-phase modulated signals under repetition-like
codes. This technique is motivated by a need to transmit packetized
or framed data bursts in high noise regimes where many powerful, short-length
codes are ineffective. In channels with strong noise, we mathematically
study the asymptotic bit error rates of this combined approach and
quantify the performance improvement over performing demodulation
and decoding separately as the code rate increases. In this context,
we also discuss a simple variant of repetition coding involving pseudorandom
code words, based on direct-sequence spread spectrum methods, that
preserves the spectral density of the encoded signal in order to maintain
resistance to narrowband interference. We describe numerical simulations
that demonstrate the advantages of this approach as an inner code
which can be used underneath modern coding schemes in high noise environments.
\end{abstract}

\section{Introduction\label{SecIntro}}

The traditional approach to designing a communications system, dating
back to Shannon's time, was to treat coding and modulation as separate
procedures that can each be studied and optimized individually. However,
from a signal detection viewpoint, it is more natural to think of
them as a single, unified process. This principle was first exploited
by Ungerboeck with the method of trellis-coded modulation (TCM) \cite{PDL90,Un82},
which sparked considerable activity in the development of such schemes,
known as \textit{waveform coding} or \textit{coded modulation} \cite{AAS86,AS91}.
TCM has seen widespread use in applications such as phone-line modems
and many other waveform coding approaches have also been proposed,
based on variations of TCM \cite{Yi96} as well as other types of
waveforms such as wavelet elements \cite{JM05,Li97} or principal
components \cite{MA03}. Many modern implementations use a combination
of the approaches and iterate decoding and demodulation in order to
be compatible with interleavers \cite{MA01,NS01}.\\

In this paper, we study a binary continuous-phase modulation (CPM),
packetized communications system in an environment with high noise
and/or strong narrowband interference (NBI) at unknown frequencies.
We are interested in noise regimes where the raw, symbol error rate
(SER) is very high, such as $0.3$ or higher. In such environments,
many forms of modern, powerful forward error correction (FEC) such
as short-length turbo or other convolutional-based codes either have
a minimal or negative effect on error rates compared to weak codes
such as simple repetition codes \cite[p. 464-465]{Pr00}, or have
long block lengths that are not suitable for use with short data frames.
For these reasons, we examine CPM based on a simple $(N,1)$ repetition
code in an additive Gaussian white noise channel. Demodulating the
resulting signal on a single-symbol basis is generally not optimal,
but is often done in practice due to other constraints such as the
presence of an interleaver in the system. On the other hand, if we
consider the resulting signal as a waveform code and combine the decoding
and demodulation into a single correlation classifier, we would expect
to get a significant improvement in the bit error rate (BER). The
main goal of this paper is to mathematically analyze this improvement
as the rate $N$ increases, and to precisely quantify the difference
in BERs between the approaches.\\

The paper establishes asymptotic bounds that compare the performance
of demodulating this waveform code on a single-codeword basis with
that of demodulation on a single-symbol basis, followed by a hard-decision
decoder to unravel the block code. For a fixed and sufficiently high
noise power $\sigma^{2}$, as the code redundancy $N\to\infty$, we
quantify the differences in the error probabilities between the combined
and separated demodulation approaches under both coherent and noncoherent
demodulation methods. This type of result is different from most asymptotic
performance bounds in the coding theory literature, which let $\sigma^{2}\to0$
for fixed code rates $N$ \cite[Ch. 8]{Pr00}, but it gives insight
into the nature of the code at high noise levels. In numerical simulations,
we find that the combined approach can drive down the error rates
by an order-of-magnitude over the separated approach. Similar techniques
for increasing the distances between CPM waveforms and reducing BERs
with such a combined approach have been studied in previous works
(e.g. \cite{Fo94,GBD07,MP84}), but only numerically for specific
block codes at fixed rates, as opposed to the asymptotic, theoretical
results we develop in this paper. Our results justify the use of repetition-like
CPM waveform codes for remote command channels and other applications
where the background noise is strong and small packet sizes prevent
the use of long-length codes, but where we also have a lot of room
to reduce the data rate using FEC. Once the error rate has been reduced
to an acceptable level, around $10^{-2}$ or $10^{-3}$, this simple
waveform code can be concatenated by more powerful FEC methods that
take over and bring it down further.\\

In this context, we also discuss a technique called ``spread coding''
to preserve the spectral density of the encoded signal and avoid periodicities
that result in spectral spikes. This approach is simply a version
of direct-sequence spread spectrum (DSSS) methods and consists of
encoding the data stream with a predetermined sequence of length-$N$
code words. These code words are known at both the transmitter and
receiver ends and can be generated in a variety of ways, such as pseudorandomly
or by a maximum-length shift register. DSSS allows a signal to maintain
a flat spectral density and increases the robustness of the system
to NBI (\cite{AAS86,La88,LL94}), in contrast to using a straight
repetition code or by simply reducing the symbol rate, and is important
in situations where interference mitigation techniques such as frequency-hopping
are not usable. However, in contrast to standard implementations and
uses of DSSS \cite{LL94,MCB13}, we think of spread coding as an FEC
method and study its effect on reducing the BER. In comparison with
more sophisticated waveform or block coding schemes, the DSSS-based
spread coding has the advantage of being compatible with existing,
uncoded CPM systems using standard hardware on the transmitter side,
and the short block lengths are compatible with small packets and
result in high speed, low complexity demodulation algorithms at the
receiver end.\\

This paper is organized as follows. Section \ref{SecSystem} discusses
the communication system and the rationale for our design choices
in more detail. The main theoretical results of the paper are stated
and described in Section \ref{SecPerf}. In Section \ref{SecSim},
we run some simulations that confirm and extend these results numerically
and discuss comparisons with other coding schemes under similar noise
environments. Appendix \ref{SecBG} reviews some background material
on statistical signal classification, and the proofs of the theorems
in Section \ref{SecPerf} are developed in Appendix \ref{SecProofs}.

\section{System description\label{SecSystem}}

Continuous-phase modulation is an effective transmission mechanism
in bandwidth- or power-limited environments. It maintains a constant
power envelope and can thus be used with nonlinear amplifiers. It
is also relatively robust to local oscillator drift at the transmitter
and is well suited for carrier tracking algorithms that can mitigate
this. However, the optimal demodulator structure at the receiver is
relatively complex due to the need to account for the memory in the
modulation. We follow the treatment in \cite[p. 185-201]{Pr00} and
consider a real-valued, binary CPM signal at a known carrier frequency
that the receiver picks up. If $B=\{b_{k}\}_{-\infty\leq k\leq L}$
is a sequence of binary symbols to be modulated, then the CPM signal
has the form
\[
s_{B}(t)=\sqrt{\frac{2E}{T}}\cos\bigg(2\pi\bigg(\theta+\omega_{c}t+\frac{h}{2}\times
\]
\begin{equation}
\sum_{k=-\infty}^{N}(2b_{k}-1)\int_{-\infty}^{t}F(s-kT)ds\bigg)\bigg),\label{CPM}
\end{equation}

where $E$ is the signal energy, $T$ is the time interval for one
symbol, $\omega_{c}$ is the carrier frequency, $\theta$ is the initial
phase, $h$ is the modulation index and $F$ is a nonnegative frequency
shaping function with $\left\Vert F\right\Vert _{L^{1}(\mathbb{R})}=1$.
Alternatively, we can also start from a complex baseband version of
(\ref{CPM}), which would lead to equivalent results. For example,
the Gaussian frequency-shift keying (GFSK) modulation scheme is defined
by the shaping function
\begin{eqnarray}
F(t) & = & \frac{\mathrm{1}}{2}\mathrm{erf}\bigg(\frac{2\pi\mathrm{BT}}{\sqrt{2\log2}}\left(t+\frac{1}{2}\right)\bigg)-\nonumber \\
 &  & \frac{1}{2}\mathrm{erf}\bigg(\frac{2\pi\mathrm{BT}}{\sqrt{2\log2}}\left(t-\frac{1}{2}\right)\bigg),\label{GFSK}
\end{eqnarray}

where BT, the bandwidth-time product, is a fixed parameter. This particular
scheme has advantages in NBI-limited environments due to its flat
spectral density shape and in-band spectral efficiency, and it is
used with $h=0.5$ and $\mathrm{BT}=0.3$ in several well-known communications
protocols such as GSM and Bluetooth. For the rest of the section,
we set $E=T=1$ to simplify the notation. We also make the approximation
that $\omega_{c}\to\infty$, which is effectively a rigorous form
of the standard ``narrowband assumption'' that sufficiently high
frequencies are filtered out in the receiver's basebanding process
(including in particular, any cross-term components at $2\omega_{c}$).
CPM signals can be expressed in several alternate but equivalent forms,
such as linear combinations of PAM signals \cite{MM95} or the outputs
of a time-invariant system \cite{Ri88}, but the standard form (\ref{CPM})
will be convenient for our purposes.\\

As discussed in Section \ref{SecIntro}, we consider simple repetition-like
codes for the symbol sequence $B$, which are motivated by a scenario
where short data packets are transmitted and need to be decoded in
real time. Command channel uplinks for satellites and unmanned aerial
vehicles often send data in this manner, and in some cases require
the use of codes with small block or constraint lengths that fit into
short packets and can be rapidly decoded by the receiver. Block length
constraints also appear in situations when the data is transmitted
continuously, but the symbols need to fit into an existing framing
structure with short data frames. In general, any repetition-based
code has weak distance properties and would not approach the classical
Shannon bound at a given SNR, but under such constraints on the code
length, weak block codes can achieve near-optimal performance as shown
in \cite{DDP98,DDP98-2}. Using sphere-packing bounds, these papers
show that at a maximum block length of $7$ symbols, for example,
a bit-to-noise ratio $E_{b}/N_{0}$ of roughly $5\mathrm{db}$ is
needed to achieve a BER of $10^{-4}$, and simple codes such as Hamming
codes are close to optimal. Some modern turbo codes (e.g. from the
CCSDS standard \cite{CCSDS07}) perform well under very low SNR and
come close to the Shannon bound, but use long constraint lengths (of
16000 or more output symbols) and are unsuitable for use with short
data packets. Under the combined demodulation approach we consider,
the correlator bank at the receiver (either coherent or noncoherent)
checks against only the possible waveforms $s_{B}$ that can occur,
given knowledge of the possible code words $B$.\\

The symbol sequence $B$ under a spread code consists of a fixed sequence
of pseudorandom code words, each $N$ symbols long. This keeps the
spectral density shape of a given modulation scheme unchanged from
an uncoded signal, maintaining the same spectral efficiency and robustness
to NBI. This coding approach is well suited for binary CPM, where
the code words $B$ can be chosen to be complements of each other
and correspond to a ``0'' or ``1'' at any given bit position,
and the fixed distance between them allows for a tractable mathematical
analysis in Section \ref{SecPerf}. We focus on this binary case in
the rest of the paper, but the basic concept can be used with larger
alphabet sizes as well, with the code words all taken to be pseudorandom.
The use of pseudorandom code words in this fashion can be thought
of as a DSSS technique, with the encoded symbols corresponding to
DSSS ``chips,'' and has been investigated in the context of CPM
in several papers (\cite{AHK98,La88,LL94,MCB13}). However, DSSS is
typically used to enable multi-user communications or to prevent detection
by a third party (low probability of intercept), rather than as an
FEC technique that reduces the BER, and has been studied primarily
in the former context. Whereas DSSS is traditionally used to expand
the signal's spectral density at a fixed data rate, the DSSS approach
we consider in this paper keeps the spectral density unchanged and
gives improved BER performance at a reduced data rate.\\

At the receiver, demodulation of a CPM signal can be performed either
coherently or noncoherently \cite[p. 295-299]{Pr00} at any given
symbol position, depending on whether the signal's initial phase is
known and kept track of (e.g. using a phase-locked loop and a Viterbi
state estimator; see \cite{CR99,CR97,Pr00}). We analyze the BERs
of both formulations in the next section, but in either case, under
the combined demodulation approach, the receiver correlates any $N$-symbol
block of the signal against only the two possible waveforms that can
appear in that position. In general, coherent signal classification
effectively increases the signal power over an equivalent noncoherent
problem and has an effect similar to a $3\mathrm{db}$ SNR improvement
\cite{SSM05}. On the other hand, the special structure of CPM signals
allows for the design of noncoherent methods that demodulate over
several data bits at a time and approach the performance of a fully
coherent classifier \cite{CR99}, and we use one such method in the
simulations in Section \ref{SecSim}. An intermediate approach is
for the demodulator to output soft, single-symbol decisions that are
passed into the decoder, as done in serially concatenated schemes
\cite{MA01,NS01}. However, the soft decisions are typically obtained
under independence or Markov assumptions, and under channels with
NBI components, they do not preserve dependencies between consecutive
decisions and are generally not equivalent to the combined approach.
Serially concatenated schemes often use long interleavers (spanning
1000 or more symbols) between the coding and modulation in order to
remove these dependencies in practice, but such interleavers are not
suitable for use with short packets. Even under a purely white noise
channel, a soft decision demodulator cannot be expressed as a likelihood
ratio (see Appendix \ref{SecBG}), so it is difficult to obtain sharp
theoretical bounds using a signal detection framework as we do in
this paper.\\

In practice, a spread code is best used as an inner code dropped into
an existing, uncoded communications system, chosen to have a high
enough code rate to bring the SER down to $10^{-2}$ or $10^{-3}$.
It can then be concatenated with a more powerful outer code that is
effective at these lower input error rates, without changing any spectral
characteristics of the signal that other aspects of the existing system
may be built around. Note that the spread coding and combined demodulation
approach is specifically meant to operate under a high noise level
and short block length constraints, and other schemes may be more
appropriate for different system requirements on e.g. the spectral
efficiency \cite{BFC09} or receiver complexity \cite{BC07,CR97}.
We summarize the description of the entire system in Figure \ref{FigBlock},
and show an example of a coded CPM signal's phase in Figure \ref{FigPhase}.

\begin{figure}[h]
\centering{}\includegraphics[trim=1.8in 2.5in 1.8in 2.5in, clip=true, scale=0.5]{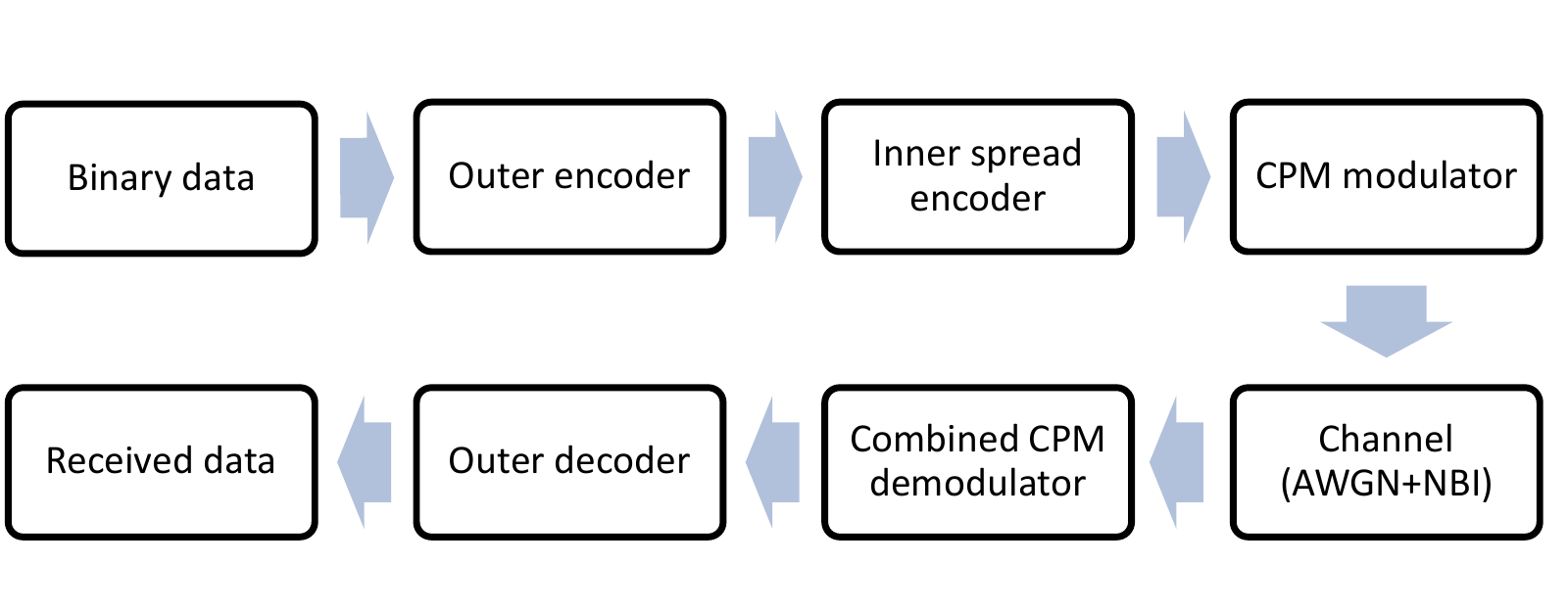}\caption{\label{FigBlock} A block diagram of the spread coded CPM system.}
\end{figure}
\begin{figure}[h]
\centering{}\includegraphics[scale=0.5]{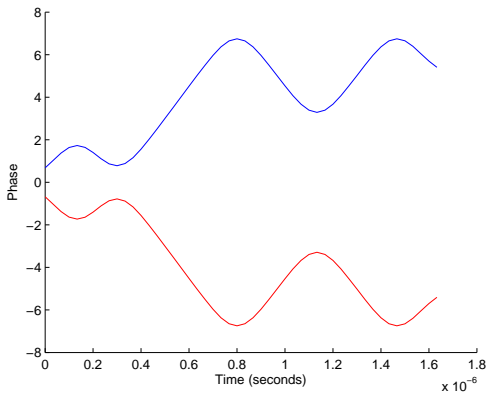}\caption{\label{FigPhase} An example GFSK signal's phase under the two possible
spread codes at a given location.}
\end{figure}

\section{Asymptotic performance of demodulation\label{SecPerf}}

We proceed to state the main results of the paper, comparing the performance
of the two demodulation approaches. The proofs of the theorems in
this section are deferred to Appendix \ref{SecProofs}. We make several
simplifying assumptions to allow for a tractable mathematical theory
of asymptotic bit error rates. We assume that the symbol sequence
$B$ is binary, so that the code words at any given bit position are
just complements of each other and have a simple distance structure.
We also assume that the code words are formed from a repetition-like
code, which results in periodic distances between CPM signals, and
that the background is additive Gaussian white noise, which leads
to explicit symbolic formulas for error probabilities in signal classification
(see Appendix \ref{SecBG}). However, our results allow for arbitrary
CPM pulse shapes, such as GFSK or raised cosine (RC) pulses.\\

Our first result establishes sharp asymptotic bounds on the distance
between two complemented binary CPM waveforms, as would be the case
in a repetition code, over an observation interval of $N$ symbols.
For any functions $f$ and $g$, we use the Landau notation $f(N)\sim g(N)$
to denote $f(N)/g(N)\to1$ as $N\to\infty$.
\begin{thm}
\noindent \label{ThmL2}For any binary symbol sequence $B$ ending
with $N$ identical symbols, let $B'$ be the same sequence with those
$N$ symbols flipped. Let $s_{B}$ and $s_{B'}$ be the corresponding
CPM signals given by (\ref{CPM}). Assume that $F$ is zero outside
$[-1,1]$ and that $F(t)=F(-t)$. As $N\to\infty,$
\begin{equation}
\lim_{\omega_{c}\to\infty}\left\Vert s_{B}-s_{B'}\right\Vert _{L^{2}(0,N)}\sim\sqrt{2N},\label{L2Bound}
\end{equation}
and for the complex-analytic signals \textup{$P^{+}s_{B}$} and $P^{+}s_{B'}$
with $\cos(\cdot)$ in (\ref{CPM}) replaced by $2^{-1/2}\exp(i\cdot)$,
\begin{equation}
\lim_{\omega_{c}\to\infty}\left\Vert P^{+}s_{B}-P^{+}s_{B'}\right\Vert _{L^{2}(0,N)}\sim\sqrt{2N},\label{L2BoundP+}
\end{equation}
\begin{equation}
\lim_{\omega_{c}\to\infty}\frac{\left|\int_{0}^{N}P^{+}s_{B}(t)\overline{P^{+}s_{B'}(t)}dt\right|}{\left\Vert P^{+}s_{B}\right\Vert _{L^{2}(0,N)}\left\Vert P^{+}s_{B'}\right\Vert _{L^{2}(0,N)}}\sim\frac{E_{0}(N)}{\sqrt{2}N},\label{IPBound}
\end{equation}
where $E_{0}$ is a function satisfying $|E_{0}(N)|\leq\frac{3}{2}\left(1+\frac{1}{h}\right)$
for all $N$.
\end{thm}
Theorem \ref{ThmL2} indicates that even when the CPM waveforms are
not chosen to be orthogonal for any fixed $N$, such as e.g. binary
MSK, they still become orthogonal in the limit as $N\to\infty$. In
practice, the condition on $F$ can be relaxed to read that $F$ is
``approximately'' zero outside $[-1,1]$, as is the case with GFSK
modulation, and the asymptotic bound still holds.\\

This result is used to compare the error probabilities of performing
demodulation and decoding jointly, denoted by $p_{\mathrm{Joint}}^{\mathrm{NC}}$
and $p_{\mathrm{Joint}}^{\mathrm{C}}$ for noncoherent and coherent
demodulation respectively, with that of doing them separately, denoted
by $p_{\mathrm{Separate}}^{\mathrm{NC}}$ and $p_{\mathrm{Separate}}^{\mathrm{C}}$.
\begin{thm}
\label{ThmCombinedDD}(Combined decoding and demodulation) Suppose
the modulation is given by (\ref{CPM}) with the constraints in Theorem
\ref{L2Bound}, and let $\omega_{c}\to\infty$. Then for a fixed noise
power $\sigma^{2}$, as $N\to\infty$,
\[
p_{\mathrm{Joint}}^{\mathrm{NC}}\sim\frac{1}{2}I_{0}\left(\frac{\sqrt{2}E_{0}(N)}{8\sigma^{2}}\right)e^{-\frac{N}{4\sigma^{2}}}\approx\frac{1}{2}e^{-\frac{N}{4\sigma^{2}}}
\]
and
\[
p_{\mathrm{Joint}}^{\mathrm{C}}\sim\frac{\sigma}{\sqrt{\pi N}}e^{-\frac{N}{4\sigma^{2}}}.
\]

\end{thm}
Here, $I_{0}$ is the modified Bessel function defined by $I_{0}(x)=\sum_{n=0}^{\infty}\frac{(x/2)^{2n}}{n!^{2}}$.
In the noncoherent bound, the coefficient turns out to be very close
to $\frac{1}{2}$ when $\sigma^{2}\geq1$. In general, the function
$E_{0}$ in (\ref{IPBound}) is oscillatory and it is not easy to
characterize its behavior precisely, but it is typically between $0$
and $1$, and since $\frac{1}{2}I_{0}(x)\sim\frac{1}{2}+\frac{1}{8}x^{2}$
as $x\to0$, the coefficient is about $\frac{1}{2}$.
\begin{thm}
\label{ThmSeparateDD}(Separate decoding and demodulation) Suppose
the modulation is given by (\ref{CPM}) with the constraints in Theorem
\ref{L2Bound}, and let $\omega_{c}\to\infty$. Assume that $N$ is
odd and $C=\lim_{\omega_{c}\to\infty}\left\Vert s_{\{0\}}-s_{\{1\}}\right\Vert _{L^{2}(0,1)}\leq\sigma$,
i.e. the noise is at least as strong as the signal separation. Then
for a fixed noise power $\sigma^{2}$ and for all sufficiently large
$N$,
\[
p_{\mathrm{Separate}}^{\mathrm{NC}}\geq\frac{\left(e^{-1/8}\left(2-e^{-1/8}\right)\right)^{\frac{N+1}{2}}}{\sqrt{2\pi N}\left(1-\frac{1}{2}e^{-1/8}\right)}\approx\frac{0.71}{\sqrt{N}}\times0.993^{N+1}
\]
and
\[
p_{\mathrm{Separate}}^{\mathrm{C}}\geq\frac{1}{\sqrt{N}}\left(1-\frac{1}{2\pi}\right)^{\frac{N+1}{2}}\approx\frac{1}{\sqrt{N}}0.917^{N+1}.
\]

\end{thm}
The distance $C$ in Theorem \ref{ThmSeparateDD} can be calculated
explicitly for certain shaping functions $F$. For binary, orthogonal
FSK with $\mathrm{MI}=1$, it is easy to check that $C=\sqrt{2}$
and that the exponents in Theorem \ref{ThmCombinedDD} become $0.882$
under the same restriction $C\leq\sigma$. These results can also
be formulated in terms of the symbol-to-noise ratio $E_{s}/N_{0}$
instead of the noise power $\sigma^{2}$, using the equivalence $E_{s}/N_{0}=\frac{1}{4\sigma^{2}}\left\Vert s_{\{0\}}-s_{\{1\}}\right\Vert _{L^{2}}^{2}$
for CPM signals \cite{AAS86}. Theorem \ref{ThmSeparateDD} holds
when $E_{s}/N_{0}\leq\frac{1}{4}=-6\,\mathrm{dB}$.\\

Theorems \ref{ThmCombinedDD} and \ref{ThmSeparateDD} together show
that as $N\to\infty$, $p_{\mathrm{Joint}}^{\mathrm{NC}}$ and $p_{\mathrm{Joint}}^{\mathrm{C}}$
decay significantly faster than $p_{\mathrm{Separate}}^{\mathrm{NC}}$
and $p_{\mathrm{Separate}}^{\mathrm{C}}$ for high noise power $\sigma^{2}$,
with a much bigger difference in the noncoherent case. Note that these
estimates are quite conservative when the demodulation is performed
over multiple data bits at a time, taking advantage of the symbol
memory inherent in CPM, but they serve to illustrate the improvement
of the combined approach over the separated approach. We also point
out that $p_{\mathrm{Joint}}^{\mathrm{C}}$ and $p_{\mathrm{Joint}}^{\mathrm{NC}}$
are of similar order and only differ by an $N^{-1/2}$ factor, while
$p_{\mathrm{Separate}}^{\mathrm{C}}$ is far smaller than $p_{\mathrm{Separate}}^{\mathrm{NC}}$.
This is intuitively reasonable, since the longer noncoherent observation
intervals in the combined approach make more use of the memory in
a CPM signal and are comparable to accounting for the entire symbol
history that led to the phase at the most recent symbol position.

\section{Extensions and simulations\label{SecSim}}

In this section, we discuss some extensions of the ideas in Section
\ref{SecPerf} and consider numerical Monte Carlo simulations of the
demodulation approaches we have studied. Some plots comparing the
spectral density of a DSSS spread coded CPM signal to that of uncoded
and repetition coded signals are shown in Figure \ref{FigSpec}. As
an illustrative example, we consider a $6$Mbit/sec GFSK signal formed
from $4000$ random, equiprobable data bits and sampled at $30$Mbit/sec,
with $N=10$, $h=0.8$ and $\mathrm{BT}=0.3$. These parameters are
similar to the GMSK modulation used in GSM, but the higher modulation
index results in a flatter spectral density shape over the main lobe.
The spectral density is estimated using the multitaper method \cite{Th82}
and the spread coded spectrum appears slightly thicker due to the
signal being ten times as long, but the overall shape is unchanged
from the uncoded case and remains flat over the main lobe. The fixed
length of the code word sequence means that the signal will still
contain periodicities, but these will be at (baseband) frequencies
that are too low to be detectable or relevant in practice. In our
simulations, we assume that the code words are generated in advance
and stored in memory at both the transmitter and receiver ends. Alternatively,
a deterministic procedure such as a maximum-length shift register
with a known initial state can be used to generate the code words
in real-time at both ends. Both implementation approaches have essentially
the same effect on the spectral density.\\

\begin{figure}[h]
\begin{centering}
\includegraphics[scale=0.5]{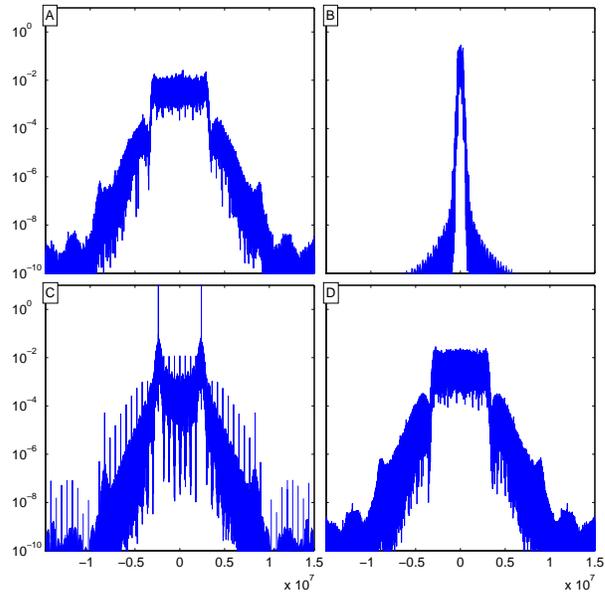}
\par\end{centering}

\caption{\label{FigSpec} Spectral density profiles with (A) no coding applied,
(B) no coding but with a reduced, $600$kbit/sec symbol rate, (C)
repetition coding with rate $N=10$, and (D) spread coding with rate
$N=10$.}
\end{figure}

In general, the distances between spread coded CPM waveforms may differ
from the repetition case addressed by Theorem \ref{ThmL2}. Except
in some special cases (e.g. when $h=1$), the periodicity techniques
used to prove Theorem \ref{ThmL2} (see Appendix \ref{SecProofs})
no longer apply and the waveforms in (\ref{L2Bound}) may no longer
be approximately orthogonal. However, on average, we find that the
results of Theorem \ref{ThmL2} still hold in practice. In Figure
\ref{FigDist}, we numerically examine the distances for the modulation
parameters discussed earlier at different rates $N$. We take all
$2^{N}$ symbol sequences $B$ of length $N$ and calculate the mean
of the distances $\left\Vert s_{B}-s_{B'}\right\Vert _{L^{2}(0,N)}$
over all such $B$. It can be seen that for large $N$, this mean
distance approaches the same $\sqrt{2N}$ limit established in Theorem
\ref{ThmL2}.\\

\begin{figure}[h]
\begin{centering}
\includegraphics[scale=0.6]{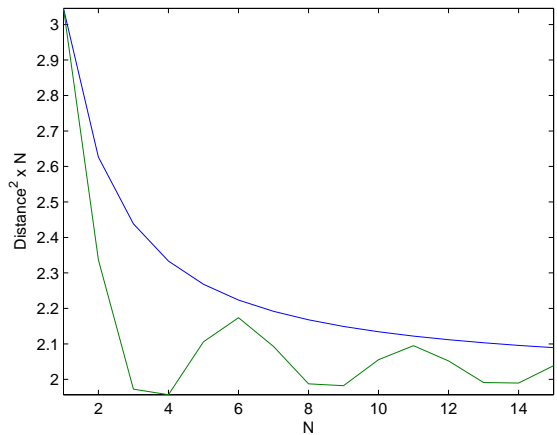}
\par\end{centering}

\caption{\label{FigDist} Mean distances over all $2^{N}$ spread coded waveforms
(blue) at different rates $N$, compared to repetition coded waveforms
(green).}
\end{figure}

We now study the performance of spread coding in simulations and compute
the error rates for a range of code rates $N$. For the purposes of
simulations, we use the method of \textit{noncoherent-block }CPM demodulation
over $K=5$ data bits (see \cite{AAS86} and \cite[p. 298-299]{Pr00})
as a substitute for a fully coherent demodulator. At each bit position
$a$, this demodulator takes a $K$-bit length of the received signal
centered at $a$ and computes envelope correlators (see Appendix \ref{SecBG})
against all $2^{K}$ possible waveforms, outputting the bit $b$ at
$a$ based on the average correlation over all waveforms containing
$b$ at $a$. For $K\leq7$, this approach exhibits performance similar
to the optimal coherent (Viterbi-based) receiver, with rapidly diminishing
gains for larger $K$, and is often used in practice for its low complexity
and simple implementation. For such a demodulator, it is generally
no longer possible to obtain precise asymptotic bounds of the type
we developed in Section \ref{SecPerf} for the special case $K=1$.
However, any block code can be incorporated into this scheme by observing
the signal over $NK$ symbols at a time but correlating it against
only the $2^{K}$ possible waveforms, using knowledge of the underlying
code (e.g. the pseudorandom code words). Hardware implementations
of this scheme typically use $K=5$ or $7$, which allows the demodulation
to be performed in close to real time. We use this approach in the
results that follow.\\

In Figure \ref{FigBER}, we use the same CPM parameters as before
and examine the BER for different code rates $N$ at the noise levels
given by $E_{s}/N_{0}=0\,\mathrm{dB}$ and $-3\,\mathrm{dB}$, with
a uniformly distributed initial phase. We also compare the results
with several standard convolutional codes of rates $1/N$ and constraint
length $8$ \cite[p. 492-494]{Pr00}. The results show that the spread
code with combined, noncoherent-block demodulation greatly outperforms
the other methods, especially in the $E_{s}/N_{0}=-3\,\mathrm{dB}$
case. At low $N$, the convolutional codes actually increase the BER
over an uncoded signal and highlight the drawbacks of powerful, short-length
codes in such noise environments, although as either $N$ or $E_{s}/N_{0}$
increase, we would expect such codes to eventually outperform the
repetition-based spread codes. For moderate rates $N$, these results
indicate that the spread codes have an ``imperfectness'' of $1-2$db
under the combined demodulation approach, in terms of the lower bounds
for codes of such rates discussed in \cite{DDP98}, so it is not possible
for a coding scheme to perform substantially better without using
larger block lengths. We also note that the BERs obtained here are
comparable to trellis coded noncoherent CPM schemes discussed in \cite{Yi96}
or the baseline serially concatenated schemes described in \cite{MA01}
(the latter paper demonstrates much better BERs by inserting interleavers
that are $2000$ or more symbols long, which effectively increases
the block length and would be unsuitable for our design constraints
as discussed in Section \ref{SecBG}).

\begin{figure}[h]
\centering{}\includegraphics[scale=0.65]{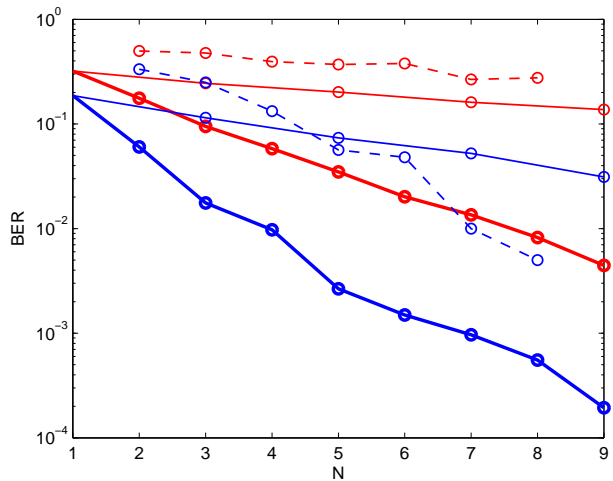}\caption{\label{FigBER} BER performance of CPM signals with spread coding
using combined demodulation (bold curves) and separated demodulation
(thin curves), along with convolutional coding (dashed curves), at
$E_{s}/N_{0}=0\,\mathrm{dB}$ (blue) and $E_{s}/N_{0}=-3\,\mathrm{dB}$
(red).}
\end{figure}

We next compare the performance of a fixed rate spread code at different
$E_{s}/N_{0}$, as well as in the presence of narrowband interference.
We simulate a single $3$kbit/sec quadrature phase-shift keyed (QPSK)
interferer at a random (but fixed) frequency within the band $\omega_{c}\pm30$Mhz,
which corresponds to the GFSK signal's main lobe (see Figure \ref{FigSpec}),
and with the same power as an equivalent level of white noise determined
by the $E_{s}/N_{0}$. We take $N=6$ and plot the BERs of a spread
coded CPM signal. We also consider the effects of changing the signal
modulation to QPSK at the same data rate. The basic effect of the
spread code increasing the distances between waveforms holds for QPSK
as well, with the same $\sqrt{2N}$ asymptotic behavior as in Theorem
\ref{L2Bound} for large $N$. Finally, we consider an example concatenating
a rate $3$ spread coded CPM signal with a rate $1/2$ convolutional
code at different error rates under white noise, which is how this
technique is best applied in practice. The results are all shown in
\ref{FigEsN0}.

\begin{figure}[h]
\centering{}\includegraphics[scale=0.6]{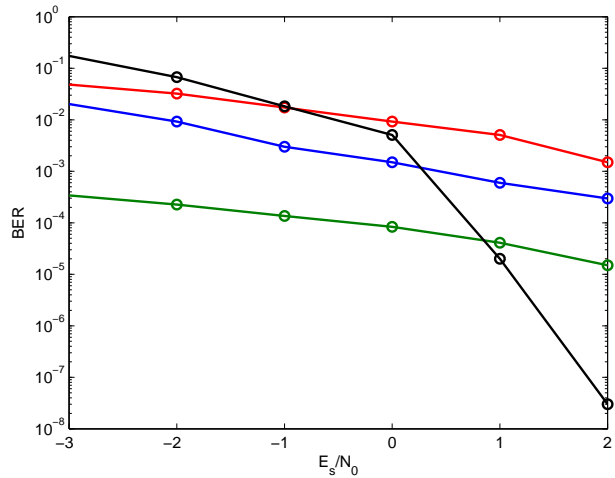}\caption{\label{FigEsN0} Performance of rate $6$ spread coded CPM under white
noise (blue) and narrowband interference (green) using the combined
demodulation approach. Also shown are rate $6$ spread coded QPSK
under white noise (red) and rate $3$ spread coded and rate $1/2$
convolutional coded CPM under white noise (black).}
\end{figure}

\section{Conclusion\label{SecEnd}}

We have proved that combining the processes of decoding and demodulation
with simple, DSSS-based and repetition-like coding schemes can confer
significant advantages in high noise environments. We expect our results
to generalize to broader classes of short-length block codes or other
modulation schemes, as well as other demodulation approaches such
as those that output soft symbol decisions and/or have observation
intervals of multiple data bits. These topics will be pursued in future
work.

\section*{Acknowledgments}

The author would like to thank Dr. Dave Colella, Dr. Richard Orr and
the anonymous referees for many valuable comments on this paper.

\begin{appendices}

\section{Background on signal classification theory\label{SecBG}}

Before proving the Theorems \ref{ThmL2}-\ref{ThmSeparateDD}, we
first review some standard results on continuous-time binary signal
classification from \cite{Po94} and \cite{VT01}. Let $f\in L^{2}(I)$
be a continuous, real deterministic function over some time interval
$I$ and let $G$ be a Gaussian white noise process with power $\sigma^{2}$.
Suppose we receive the (random) signal $Y(t)$ and want to determine
which of the hypotheses $H_{0}=\{Y=G\}$ and $H_{1}=\{Y=f+G\}$ holds.
The likelihood ratio test for this \textit{coherent classification}
problem is given by
\begin{equation}
e^{\frac{1}{\sigma^{2}}\left\langle f,Y\right\rangle _{I}-\frac{d^{2}}{2}}\lessgtr\tau.\label{Detector1}
\end{equation}
where $\tau$ is a fixed threshold, $d=\frac{1}{\sigma}\left\Vert f\right\Vert _{L^{2}}$
and the inner product $\left\langle \cdot,Y\right\rangle _{I}$ is
interpreted as a white noise functional. The performance of the test
is given in terms of the error function $\mathrm{erf}(x)=2\pi{}^{-1/2}\int_{0}^{x}e^{-y^{2}}dy$
by
\begin{eqnarray*}
P(H_{1}|H_{1}) & = & \frac{1}{2}\mathrm{erf}\left(\frac{1}{\sqrt{2}}\left(\frac{\log\tau}{d}+\frac{d}{2}\right)\right)+\frac{1}{2},\\
P(H_{1}|H_{0}) & = & \frac{1}{2}\mathrm{erf}\left(\frac{1}{\sqrt{2}}\left(\frac{\log\tau}{d}-\frac{d}{2}\right)\right)+\frac{1}{2}.
\end{eqnarray*}

We take $\tau=1$ and $f=f_{1}-f_{2}$, where $f_{1}$ and $f_{2}$
are known, modulated waveforms corresponding to a $0$ or $1$ bit.
The test (\ref{Detector1}) reduces to a simple correlation classifier:
\begin{equation}
\left\langle f_{1},Y\right\rangle _{I}\lessgtr\left\langle f_{2},Y\right\rangle _{I}.\label{Detector2}
\end{equation}
Assuming that $P(H_{0})=P(H_{1})$, the probability of a symbol classification
error is
\begin{eqnarray}
p & = & \frac{1}{2}-\frac{1}{2}\mathrm{erf}\left(\frac{\left\Vert f_{1}-f_{2}\right\Vert _{L^{2}(I)}}{2\sqrt{2}\sigma}\right).\label{PCoherent}
\end{eqnarray}
Similar results can also be shown for \textit{noncoherent classification},
where $f$ is now complex-valued and we instead have $H_{1}=\{Y=2^{1/2}\mathrm{Re}(e^{2\pi i\theta}f)+G\}$
for some unknown phase angle $\theta$. If $\theta$ is assumed to
be random and uniformly distributed over $[0,1]$, the correlation
classifier (\ref{Detector2}) is replaced by the \textit{envelope
classifier} given by
\[
\left|\left\langle f_{1},Y\right\rangle _{I}\right|\lessgtr\left|\left\langle f_{2},Y\right\rangle _{I}\right|.
\]

When $f_{1}$ and $f_{2}$ have equal norms, the error probability
of this classifier is \cite[p. 311]{Pr00}
\begin{eqnarray}
p & = & e^{-\frac{1}{2}(u^{2}+v^{2})}\sum_{k=0}^{\infty}(u/v)^{k}I_{k}(uv)\nonumber \\
 &  & -\frac{1}{2}e^{-\frac{1}{2}(u^{2}+v^{2})}I_{0}(uv),\label{PNonCoh}
\end{eqnarray}
where
\begin{eqnarray*}
u & = & \frac{\left\Vert f_{1}-f_{2}\right\Vert _{L^{2}(I)}}{2\sqrt{2}\sigma}\left(1-\sqrt{1-\left|R_{I}(f_{1},f_{2})\right|^{2}}\right)^{1/2},\\
v & = & \frac{\left\Vert f_{1}-f_{2}\right\Vert _{L^{2}(I)}}{2\sqrt{2}\sigma}\left(1+\sqrt{1-\left|R_{I}(f_{1},f_{2})\right|^{2}}\right)^{1/2},
\end{eqnarray*}

\[
R_{I}(f_{1},f_{2})=\frac{\left\langle f_{1},f_{2}\right\rangle _{I}}{\left\Vert f_{1}\right\Vert _{L^{2}(I)}\left\Vert f_{2}\right\Vert _{L^{2}(I)}},
\]

and $I_{k}(x)=\sum_{n=0}^{\infty}\frac{(x/2)^{2n+k}}{n!(n+k)!}$ is
the modified Bessel function of order $k$. The probability (\ref{PNonCoh})
is minimized when $R_{I}(f_{1},f_{2})=0$ \cite{SBZ10}, in which
case (\ref{PNonCoh}) simplifies to
\begin{equation}
p=\frac{1}{2}e^{-\frac{1}{8\sigma^{2}}\left\Vert f_{1}-f_{2}\right\Vert _{L^{2}(I)}^{2}}.\label{PNonCoh2}
\end{equation}

We also note that the error function in (\ref{PCoherent}) satisfies
the elementary bounds
\begin{equation}
\mathrm{erf}(x)\leq\frac{2x}{\sqrt{\pi}}\label{Erf1}
\end{equation}
for $x\geq0$ and
\begin{equation}
\mathrm{erf}(x)\sim1-\frac{e^{-x^{2}}}{\sqrt{\pi}x}\label{Erf2}
\end{equation}
as $x\to\infty$.

\section{Proof of Theorems \ref{ThmL2}-\ref{ThmSeparateDD}\label{SecProofs}}
\begin{proof}[Proof of Theorem \ref{ThmL2}]
Let $s_{B,\theta_{0}}(t)$ be (\ref{CPM}) with the phase $\theta$
replaced by $\theta+\theta_{0}$, where $\theta_{0}$ is either $0$
or $\frac{1}{4}$. The purpose of this extra parameter will become
clear later. We also define $J(N,t)=\sum_{k=0}^{N}\int_{-\infty}^{t}F(s-k)ds$
to simplify some notation. First, we have
\begin{eqnarray*}
 &  & \left\Vert s_{B,\theta_{0}}-s_{B'}\right\Vert _{L^{2}(0,N)}^{2}\\
 & = & \left\Vert s_{B,\theta_{0}}\right\Vert _{L^{2}(0,N)}^{2}+\left\Vert s_{B'}\right\Vert _{L^{2}(0,N)}^{2}-2\left\langle s_{B,\theta_{0}},s_{B'}\right\rangle _{(0,N)}\\
 & = & 2N-E_{1}(\omega_{c})-2\int_{0}^{N}\cos\bigg(2\pi\bigg(\theta_{0}+\frac{h}{2}\times\\
 &  & \sum_{k=-\infty}^{N}(2b_{k}-1-2b_{k}'+1)\int_{-\infty}^{t}F(s-k)ds\bigg)\bigg)dt
\end{eqnarray*}
\begin{eqnarray}
 & = & 2\int_{0}^{N}\left(1-\cos\left(2\pi\left(\theta_{0}+h\, J(N,t)\right)\right)\right)dt\nonumber \\
 &  & \quad-E_{1}(\omega_{c}),\label{PhaseDiff}
\end{eqnarray}
where $E_{1}(\omega_{c})=O(1/\omega_{c})$. Taking $\omega_{c}\to\infty$
causes this term to drop out. Now the function $1-\cos(2\pi x)$ is
well approximated by a sum of hat functions $2\sum_{m=-\infty}^{\infty}H(x-m)$,
where the hat function $H(x)$ is given by $H(x)=2x$ for $x\in[0,\frac{1}{2})$,
$H(x)=2-2x$ for $x\in[\frac{1}{2},1)$, and $H(x)=0$ otherwise.
We define the approximation error
\[
E_{2}(x)=2\sum_{m=-\infty}^{\infty}H(x-m)-(1-\cos(2\pi x)).
\]
$E_{2}(x)$ is periodic with period $1$ and satisfies $|E_{2}(x)|\leq\frac{1}{4}$
and $E_{2}(x)=-E_{2}(\frac{1}{2}-x)=E_{2}(1-x)$ (see Figure \ref{FigApproxError}),
which implies that $\int_{0}^{M}E_{2}(x+\frac{1}{2})dx=0$ for any
$M\in\mathbb{Z}$. The conditions on $F$ show that $J(N,\cdot)$
is a continuous, nondecreasing function that equals $t+\frac{1}{2}$
for $t\in\mathbb{Z}\cap[0,N]$. This shows that as long as $N\geq\frac{1}{h}$,
there exists a point $M'\leq N$ with $\left|M'-N\right|\leq\frac{1}{h}$
such that
\[
\int_{0}^{M'}E_{2}\left(\theta_{0}+h\, J(N,t)\right)dt=0.
\]
Now let $A\leq0$ be the closest point to $0$ such that $h\, J(N,A)=0$,
and let $B\geq N-\frac{1}{h}$ be the closest point to $N$ such that
$h\, J(N,B)-\frac{1}{2}\in\mathbb{Z}$. Since $F$ is identically
zero outside $[-1,1]$, we must have $A\geq-1$ and $B\leq N+\frac{1}{2}$.
We can use this to find that as $N\to\infty$,
\begin{eqnarray*}
 &  & 2\int_{0}^{N}\left(1-\cos\left(2\pi\left(\theta_{0}+h\, J(N,t)\right)\right)\right)dt\\
 & = & 4\int_{0}^{N}\sum_{m=-\infty}^{\infty}H\left(\theta_{0}+h\, J(N,t)-m\right)dt-\\
 &  & \quad2\int_{0}^{N}E_{2}\left(\theta_{0}+h\, J(N,t)\right)dt\\
 & \sim & 4\int_{A}^{B}\sum_{m=-\infty}^{\infty}H\left(\theta_{0}+h\, J(N,t)-m\right)dt-\\
 &  & \quad2\int_{M'}^{N}E_{2}\left(\theta_{0}+h\, J(N,t)\right)dt\\
 & \sim & 4\cdot\frac{B-A}{2}+E_{0}(N)\\
 & \sim & 2N,
\end{eqnarray*}
where $E_{0}(N)=2\int_{M'}^{N}E_{2}\left(\theta_{0}+h\, J(N,t)\right)dt$
satisfies the bound $|E_{0}(N)|\leq2(-A-(B-N))+\frac{N-M'}{2}\leq\frac{3}{2}\left(1+\frac{1}{h}\right)$.
Taking $\theta_{0}=0$ completes the proof. For the complex case (\ref{L2BoundP+}),
we simply replace $\theta$ by $\theta+\frac{1}{4}$ and obtain the
same bound for the quadrature signal $s_{B}^{\star}$ with $\cos(\cdot)$
in (\ref{CPM}) replaced by $\sin(\cdot)$, and therefore also for
the analytic signal $P^{+}s_{B}=2^{-1/2}(s_{B}+is_{B}^{\star})$.
Finally, for the inner product estimate (\ref{IPBound}), we first
take $\theta_{0}=0$ to get
\begin{eqnarray*}
 &  & \mathrm{Re}\left(R_{(0,N)}(P^{+}s_{B},P^{+}s_{B'})\right)\\
 & = & \frac{\left\Vert P^{+}s_{B}\right\Vert _{L^{2}(0,N)}^{2}+\left\Vert P^{+}s_{B'}\right\Vert _{L^{2}(0,N)}^{2}}{2\left\Vert P^{+}s_{B}\right\Vert _{L^{2}(0,N)}\left\Vert P^{+}s_{B'}\right\Vert _{L^{2}(0,N)}}\\
 &  & \quad-\frac{\left\Vert P^{+}s_{B}-P^{+}s_{B'}\right\Vert _{L^{2}(0,N)}^{2}}{2\left\Vert P^{+}s_{B}\right\Vert _{L^{2}(0,N)}\left\Vert P^{+}s_{B'}\right\Vert _{L^{2}(0,N)}}\\
 & \sim & \frac{E_{0}(N)}{2N}.
\end{eqnarray*}
By taking $\theta_{0}=\frac{1}{4}$, the same result follows for $\mathrm{Im}\left(R_{(0,N)}(P^{+}s_{B},P^{+}s_{B'})\right)$,
and thus
\[
\left|R_{(0,N)}(P^{+}s_{B},P^{+}s_{B'})\right|\sim\frac{E_{0}(N)}{\sqrt{2}N}.
\]

\end{proof}
$\quad$
\begin{figure}[h]
\centering{}\includegraphics[scale=0.5]{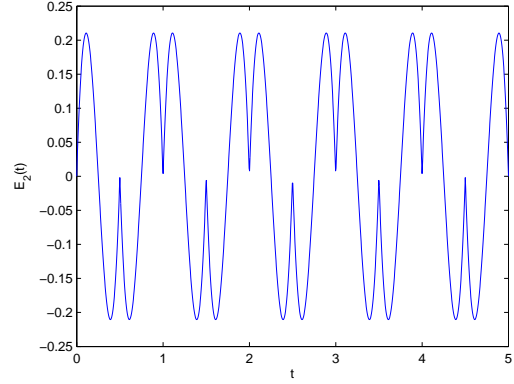}\caption{\label{FigApproxError} The approximation error $E_{2}$.}
\end{figure}

\begin{proof}[Proof of Theorem \ref{ThmCombinedDD}]
The coherent result follows immediately from combining (\ref{PCoherent})
and (\ref{L2Bound}). As $N\to\infty$, we use (\ref{Erf2}) to find
that
\begin{eqnarray*}
p_{\mathrm{Joint}}^{\mathrm{C}} & \sim & \frac{1}{2}-\frac{1}{2}\mathrm{erf}\left(\frac{\sqrt{2N}}{2\sqrt{2}\sigma}\right).\\
 & \sim & \frac{1}{2}-\frac{1}{2}\left(1-\frac{e^{-\left(\frac{\sqrt{2N}}{2\sqrt{2}\sigma}\right)^{2}}}{\sqrt{\pi}\left(\frac{\sqrt{2N}}{2\sqrt{2}\sigma}\right)}\right)\\
 & = & \frac{\sigma}{\sqrt{\pi N}}e^{-\frac{N}{4\sigma^{2}}}.
\end{eqnarray*}

The noncoherent case uses the ``almost orthogonality'' implied by
(\ref{IPBound}) to show that (\ref{PNonCoh}) is closely approximated
by (\ref{PNonCoh2}). Letting $\delta_{0,k}=1$ for $k=0$ and $\delta_{0,k}=0$
otherwise, we use (\ref{L2BoundP+}) and (\ref{IPBound}) in (\ref{PNonCoh})
to get
\begin{eqnarray}
 &  & p_{\mathrm{Joint}}^{\mathrm{NC}}\nonumber \\
 & = & e^{-\frac{1}{2}(u^{2}+v^{2})}\sum_{k=0}^{\infty}\left(1-\frac{\delta_{0,k}}{2}\right)(u/v)^{k}I_{k}(uv)\\
 & = & e^{-\frac{1}{8\sigma^{2}}\left\Vert P^{+}s_{B}-P^{+}s_{B'}\right\Vert _{L^{2}(0,N)}^{2}}\sum_{k=0}^{\infty}\left(1-\frac{\delta_{0,k}}{2}\right)\times\nonumber \\
 &  & \quad\left(\frac{1-\sqrt{1-\left|R_{(0,N)}(P^{+}s_{B},P^{+}s_{B'})\right|^{2}}}{1+\sqrt{1-\left|R_{(0,N)}(P^{+}s_{B},P^{+}s_{B'})\right|^{2}}}\right)^{k/2}\times\nonumber \\
 &  & \quad I_{k}\left(\frac{1}{8\sigma^{2}}\frac{\left\Vert P^{+}s_{B}-P^{+}s_{B'}\right\Vert _{L^{2}(0,N)}^{2}}{\left|R_{(0,N)}(P^{+}s_{B},P^{+}s_{B'})\right|^{-1}}\right)\nonumber \\
 & \sim & \sum_{k=0}^{\infty}\left(1-\frac{\delta_{0,k}}{2}\right)\left(\frac{1-\sqrt{1-E_{0}(N)/(4N^{2})}}{1+\sqrt{1-E_{0}(N)/(4N^{2})}}\right)^{k/2}\nonumber \\
 &  & \quad\times I_{k}\left(\frac{\sqrt{2}E_{0}(N)}{8\sigma^{2}}\right)e^{-\frac{N}{4\sigma^{2}}}.\label{BigSum}
\end{eqnarray}
Using the basic property $\left|I_{k}(x)\right|\geq\left|I_{k+1}(x)\right|$
for each $k\geq1$ \cite{WW27}, we conclude that as $N\to\infty$,
all terms in the sum (\ref{BigSum}) go to zero uniformly in $k$
except for the $k=0$ term, which goes to $\frac{1}{2}I_{0}\left(\frac{\sqrt{2}E_{0}(N)}{8\sigma^{2}}\right)e^{-\frac{N}{4\sigma^{2}}}$.
Consequently, we end up with
\begin{eqnarray*}
p_{\mathrm{Joint}}^{\mathrm{NC}} & \sim & \frac{1}{2}I_{0}\left(\frac{\sqrt{2}E_{0}(N)}{8\sigma^{2}}\right)e^{-\frac{N}{4\sigma^{2}}}.
\end{eqnarray*}

\end{proof}
$\quad$
\begin{proof}[Proof of Theorem \ref{ThmSeparateDD}]
If $p$ is the error probability in demodulating an individual symbol,
the error probability in a majority vote decoder is given by
\begin{equation}
p'=\sum_{k=(N+1)/2}^{N}\frac{N!}{k!(N-k)!}p^{k}(1-p)^{N-k}.\label{MajVote}
\end{equation}
For large $N$, we can approximate this using a classical extension
of the Stirling expansion for the gamma function,
\[
\Gamma(N+k+1)=\sqrt{2\pi}N^{N+k+\frac{1}{2}}e^{-N+E_{3}(N,k)},
\]
with the error $E_{3}$ satisfying
\[
\frac{1}{N+1}\leq\frac{E_{3}(N,k)}{\frac{k^{2}}{2}-\frac{k}{2}+\frac{1}{12}}\leq\frac{1}{N}
\]
for $N\geq k>0$ \cite{WW27}. This can be applied to (\ref{MajVote})
to obtain a simple geometric series after some simplification:
\begin{eqnarray}
 &  & p'\nonumber \\
 & = & \sum_{k=0}^{(N-1)/2}\frac{\Gamma(N+1)}{\Gamma(\frac{N+1}{2}+k+1)\Gamma(\frac{N-1}{2}-k+1)}\times\nonumber \\
 &  & \quad p^{\frac{N+1}{2}+k}(1-p)^{\frac{N-1}{2}-k}\nonumber \\
 & \sim & \sum_{k=0}^{(N-1)/2}\frac{N^{N+\frac{1}{2}}e^{-N}p^{\frac{N+1}{2}+k}(1-p)^{\frac{N-1}{2}-k}}{(\frac{N+1}{2})^{\frac{N+1}{2}+k+\frac{1}{2}}(\frac{N-1}{2})^{\frac{N-1}{2}-k+\frac{1}{2}}e^{-N}}\nonumber \\
 &  & \quad\times(2\pi)^{-1/2}e^{E_{3}(N,k)-E_{3}(\frac{N+1}{2},k)-E_{3}(\frac{N-1}{2},-k)}\nonumber \\
 & \sim & \frac{N^{N+\frac{1}{2}}p^{\frac{N+1}{2}}(1-p)^{\frac{N-1}{2}}}{\sqrt{2\pi}(\frac{N+1}{2})^{\frac{N+1}{2}+\frac{1}{2}}(\frac{N-1}{2})^{\frac{N-1}{2}+\frac{1}{2}}}\times\nonumber \\
 &  & \quad\sum_{k=0}^{\infty}\left(\frac{4p}{\left(N^{2}-1\right)(1-p)}\right)^{k}\times\nonumber \\
 &  & \quad e^{-\frac{1}{N}\left(\frac{(3N^{2}+1)(k^{2}+1/6)+(N^{2}+4N-1))k}{2(N^{2}-1)}\right)}\label{ExpoSum}\\
 & \sim & \frac{(2N)^{N+\frac{1}{2}}(p(1-p))^{\frac{N+1}{2}}}{(N^{2}-1)^{\frac{N}{2}}((1-2p)N+1)\sqrt{\pi}}\nonumber \\
 & = & \frac{(4p(1-p))^{\frac{N+1}{2}}}{\sqrt{2\pi}((1-2p)\sqrt{N}+1/\sqrt{N})},\label{ApproxSum}
\end{eqnarray}
where in (\ref{ExpoSum}), we used the fact that the exponential factor
in the sum tends uniformly to $1$ for all $k$. For the noncoherent
case, we obtain a lower bound on the overall BER by assuming $\left\langle P^{+}s_{\{0\}},P^{+}s_{\{1\}}\right\rangle =0$,
so that (\ref{PNonCoh2}) applies, and using (\ref{L2BoundP+}) and
the condition $C\leq\sigma$. The function $e^{-x}(2-e^{-x})$ is
decreasing for $x>0$, so for sufficiently large $N$,
\begin{eqnarray*}
 &  & p_{\mathrm{Separate}}^{\mathrm{NC}}\\
 & \geq & \frac{\left(4\left(\frac{1}{2}e^{-C^{2}/(8\sigma^{2})}\right)\left(1-\frac{1}{2}e^{-C^{2}/(8\sigma^{2})}\right)\right)^{\frac{N+1}{2}}}{\sqrt{2\pi N}\left(1-\frac{1}{2}e^{-C^{2}/(8\sigma^{2})}\right)}\\
 & \geq & \frac{\left(e^{-1/8}\left(2-e^{-1/8}\right)\right)^{\frac{N+1}{2}}}{\sqrt{2\pi N}\left(1-\frac{1}{2}e^{-1/8}\right)}.
\end{eqnarray*}
In the coherent case, the error probabilities $p$ at consecutive
symbols from the demodulator are no longer independent, but we can
again find a lower bound on the BER by considering a best-case scenario
where at any symbol position, all previous symbols were demodulated
correctly and the initial phase is thus known. For sufficiently large
$N$, we use the estimate (\ref{L2Bound}) with (\ref{ApproxSum})
to get
\begin{eqnarray*}
 &  & p_{\mathrm{Separate}}^{\mathrm{C}}\\
 & \geq & \frac{\left(4\left(\frac{1}{2}-\frac{1}{2}\mathrm{erf}\left(\frac{C}{2\sqrt{2}\sigma}\right)\right)\left(-\frac{1}{2}+\frac{1}{2}\mathrm{erf}\left(\frac{C}{2\sqrt{2}\sigma}\right)\right)\right)^{\frac{N+1}{2}}}{\sqrt{2\pi N}\mathrm{erf}\left(\frac{C}{2\sqrt{2}\sigma}\right)}\\
 & = & \frac{1}{\sqrt{2\pi N}\mathrm{erf}\left(\frac{C}{2\sqrt{2}\sigma}\right)}\left(1-\mathrm{erf}\left(\frac{C}{2\sqrt{2}\sigma}\right)^{2}\right)^{\frac{N+1}{2}}.
\end{eqnarray*}
Since $C\leq\sigma$, we take the inequality (\ref{Erf1}) into account
and obtain
\begin{eqnarray*}
 &  & p_{\mathrm{Separate}}^{\mathrm{C}}\\
 & \geq & \frac{1}{\sqrt{2\pi N}\left(\frac{C}{\sqrt{2\pi}\sigma}\right)}\left(1-\left(\frac{C}{\sqrt{2\pi}\sigma}\right)^{2}\right)^{\frac{N+1}{2}}\\
 & \geq & \frac{1}{\sqrt{N}}\left(1-\frac{1}{2\pi}\right)^{\frac{N+1}{2}}.
\end{eqnarray*}

\end{proof}
\end{appendices}

\bibliographystyle{plain}
\bibliography{Fullbib}

\end{document}